\documentclass[reqno,12pt]{article} %twoside is a good option.
\usepackage{url}
\usepackage{amsmath}
\usepackage{amssymb}
\usepackage{fancybox}

\usepackage[top=1in, bottom=1in, left=1in, right=1in]{geometry} %Controls margins

\newtheorem{theorem}{Theorem}[section]

\newtheorem{lemma}{Lemma}[section]

\newcommand{\qed}{\hfill\rule{2.1mm}{2.1mm}}

\newcommand{\R}{\mathbf{R}}

\newcommand{\eref}[1]{$(\ref{#1})$}

\newcommand\e{{\bf e}}
\newcommand\g{{\bf g}}
\renewcommand\u{{\bf u}}
\renewcommand\v{{\bf v}}
\newcommand\x{{\bf x}}
\newcommand\y{{\bf y}}
\newcommand\bz{{\bf 0}}
\newcommand\eps{\epsilon}

\newcommand{\rank}{\mathrm{rank}\,}
\newcommand{\st}{\mathrm{s.t.}\;}

\newcommand{\ra}{\rightarrow}

\title{Nuclear norm minimization for the planted clique and biclique
problems\thanks{Supported in part by a Discovery Grant from
NSERC (Natural Science and Engineering Research Council of Canada)}}
\author{Brendan P.W.~Ames\thanks{Department of
Combinatorics and Optimization, University of Waterloo, 200
University Avenue W., Waterloo, Ontario N2L 3G1, Canada,
bpames@math.uwaterloo.ca} \and
Stephen A.~Vavasis\thanks{Department of
Combinatorics and Optimization, University of Waterloo, 200
University Avenue W., Waterloo, Ontario N2L 3G1, Canada,
vavasis@math.uwaterloo.ca}}

\begin{document}
\maketitle

\begin{abstract}
We consider the problems of finding a maximum clique in a graph and finding
a maximum-edge biclique in a bipartite graph.  Both 
problems are NP-hard.  We
write both problems as matrix-rank minimization and then relax them using
the nuclear norm.  This technique,
which may be regarded as a generalization of compressive sensing,
has recently been shown to be
an effective way to solve rank optimization problems.
In the special cases that
the input graph has a planted clique or biclique (i.e., a single large clique
or biclique plus diversionary edges), our algorithm 
successfully provides 
an exact solution to the original instance. For each problem, we provide
two analyses of when our algorithm succeeds.
In the first analysis, the diversionary edges are placed
by an adversary.  In the second, they are placed at random.
In the case of random edges for the
planted clique problem, we obtain the same bound as Alon, Krivelevich and
Sudakov as well as Feige and Krauthgamer, but we use different techniques.
\end{abstract}

\section{Introduction}
\label{sec:intro}
Several recent papers including Recht et al.\ \cite{Recht-Fazel-Parrilo:2007} 
and Cand\`es and Recht \cite{Candes-Recht:2008} consider nuclear norm minimization as a
convex relaxation of matrix rank minimization.  {\em Matrix rank
minimization} refers to the problem of finding a matrix $X\in\R^{m\times n}$
to minimize $\rank(X)$ subject to linear constraints on $X$. 
As we shall show in
Sections~\ref{sec:maxclique} and \ref{sec:maxbiclique}, the clique
and biclique problems, both NP-hard,
are easily expressed as matrix rank minimization, thus showing that
matrix rank minimization
is also NP-hard.

Each of the two papers mentioned in the previous paragraph has results
of the following general form.  Suppose an
instance of  matrix rank minimization
is posed in which it is known {\em a priori} that a solution of very low
rank exists.  Suppose further that the constraints are random in some
sense.  Then the nuclear norm relaxation turns out to be exact, i.e., it
recovers the (unique) solution of low rank.  The {\em nuclear norm}
of a matrix $X$, also called the {\em trace norm}, is defined to be the
sum of the singular values of $X$.

These authors build upon recent breakthroughs in compressive sensing
\cite{Gilbert,Donoho,CandesRombergTao}.  In compressive sensing,
the problem is to recover a sparse vector that solves a set of linear
equations.  In the case that the equations are randomized and
a very sparse solution exists, compressive
sensing can be solved by relaxation to the $l_1$ norm.  The correspondence
between matrix rank minimization and compressive sensing is as follows:
matrix rank (number of nonzero singular values) corresponds to vector sparsity
(number of nonzero entries) and nuclear norm corresponds to $l_1$ norm.

Our results follow the spirit of Recht et al.\ but use different
technical approaches.  We establish results about two well known
graph theoretic problems, namely maximum clique and maximum-edge biclique.
The maximum clique problem takes as input an undirected graph and asks
for the largest clique (i.e., induced subgraph of nodes that are
completely interconnected).  This problem is one of Karp's original NP-hard
problems \cite{GJ}.  The maximum-edge biclique takes as input a
bipartite graph $(U,V,E)$ and asks for the subgraph that is a complete
bipartite graph $K_{m,n}$ that maximizes the product $mn$.  This problem
was shown to be NP-hard by Peeters \cite{Peeters}.

In
Sections~\ref{sec:maxclique} and \ref{sec:maxbiclique},
we relax these problems to convex optimization using the nuclear norm.
For each problem, we show that convex optimization can recover the
exact solution in two cases.  The first case,
described in Section~\ref{sec:maxcliqueadver},
is the adversarial case:
the $N$-node graph under consideration consists of a single $n$-node clique
plus a number of diversionary edges chosen by an adversary.  We show
that the algorithm can tolerate up to $O(n^2)$ diversionary edges provided
that no non-clique vertex is adjacent to more than $O(n)$ clique vertices.
We argue also that these two bounds, $O(n^2)$ and $O(n)$, are the best 
possible.
We show analogous results for the biclique problem
in Section~\ref{sec:maxbicliqueadver}.

Our second analysis,
described in Sections~\ref{sec:maxcliquerand}
and \ref{sec:maxbicliquerand},
supposes that the graph contains a single clique or
biclique, while the remaining nonclique edges are inserted independently
at random with fixed probability $p$.  This problem has been studied
by Alon et al.\ \cite{Alon} and by Feige and Krauthgamer \cite{Feige:2000}.
In the
case of clique, we
obtain the same result as they do, namely, that as long as the clique
has at least $O(N^{1/2})$ nodes, where $N$ is the number of nodes in $G$,
then our algorithm will find it.  Like Feige and Krauthgamer, our algorithm
also certifies that the maximum clique has been found due to
a uniqueness result for convex optimization, which we present
in Section~\ref{sec:maxcliqueopt}.  We believe that
our technique is more general than Feige and Krauthgamer; for example,
ours extends essentially without alteration to the biclique problem, whereas
Feige and Krauthgamer rely on some special properties of
the clique problem.  Furthemore, Feige and Krauthgamer use
more sophisticated probabilistic tools (martingales), whereas our
results use only Chernoff bounds
and classical theorems about the norms of random matrices.
The random matrix results needed for our main
theorems are presented in Section~\ref{sec:prelim}.

Our interest in the planted clique and biclique problems arises
from applications in data mining.  In data mining, one seeks a pattern
hidden in an apparently unstructured set of data.  A natural question
to ask is whether a data mining algorithm is able to find the hidden
pattern in the case that it is actually present but obscured by noise.
For example, in the realm of clustering, Ben-David \cite{BD} has shown
that if the data is actually clustered, then a clustering algorithm can
find the clusters.  The clique and biclique problems are both
simple model problems for data mining.  For example, Pardalos
\cite{Pardalos} reduces a data mining problem in epilepsy
prediction to a maximum clique problem.
Gillis and Glineur \cite{GillisGlineur} use the biclique problem
as a model problem for nonnegative matrix factorization and
finding features in images.

\section{Results on norms of random matrices}
\label{sec:prelim}
In this section we provide a few results concerning random matrices
with independently identically distributed (i.i.d.) entries of
mean 0.  In particular, the probability distribution $\Omega$ for
an entry $A_{ij}$ will be as follows:
$$
A_{ij}=\left\{
\begin{array}{ll}
1 & \mbox{with probability $p$,} \\
-p/(1-p) & \mbox{with probability $1-p$.}
\end{array}
\right.
$$
It is easy to check that the variance of $A_{ij}$ is $\sigma^2=p/(1-p)$.

We start by recalling a theorem of F\"uredi and Koml\'os 
\cite{Furedi-Komlos:1981}:
\begin{theorem}
    \label{Furedi-Komlos}
        For all integers $i,j$, $1\le j \le i \le n$, let $A_{ij}$ be distributed according to $\Omega$.
        Define symmetrically $A_{ij} = A_{ji}$ for all $i < j$.

        Then the random symmetric matrix $A = [A_{ij}]$ satisfies
        \[
            \| A \| \le 3 \sigma \sqrt{n}
    \]
    with probability at least to $1 - \exp(-c n^{1/6})$ for some 
$c > 0$ that depends on $\sigma$.
\end{theorem}

\noindent{\bf Remark 1.} In this theorem and for the rest of the
paper, $\Vert A\Vert$ denotes $\Vert A\Vert_2$, often called the
spectral norm.  It is equal to the maximum singular value of $A$ or
equivalently to the square root of the maximum eigenvalue of $A^TA$.

\noindent{\bf Remark 2.} The theorem is not stated exactly in this way
in \cite{Furedi-Komlos:1981}; the stated form of the theorem can be deduced by taking $k = (\sigma / K)^{1/3} n^{1/6}$ and $v = \sigma \sqrt{n}$ in
the inequality
$$
    P(\max |\lambda| > 2\sigma \sqrt{n} + v) < \sqrt{n} \exp(-kv /(2 \sqrt{n} + v))
$$
on p.~237.

\noindent{\bf Remark 3.} As mentioned above, the mean value of
entries of $A$ is 0.  This is crucial for the theorem; a distribution
with any other mean
value would lead to $\Vert A\Vert=O(n)$.

A similar theorem due to Geman \cite{Geman:1980}
is available for unsymmetric matrices.

\begin{theorem}
    \label{geman tail bound}
    Let $A$ be a $\lceil yn \rceil \times n$ matrix whose entries are chosen according to $\Omega$
    for fixed $y \in \R_+$.
    Then, with probability at least $1 - c_1 \exp(-c_2 n^{c_3})$
    where $c_1 > 0$, $c_2>0$, and $c_3 > 0$  depend on $p$ and $y$,
    $$
        \|A\| \le c_4 \sqrt{n}
    $$
    for some $c_4>0$ also depending on $p,y$.
\end{theorem}

As in the case of \cite{Furedi-Komlos:1981}, this theorem is not stated exactly this
way in Geman's paper, but can be deduced from the equations on
pp.~255--256 by taking $k=n^q$ for a $q$ satisfying $(2\alpha+4)q<1$.

The last theorem about random matrices requires a version of
the well known Chernoff bounds, which is as follows
    (see \cite[Theorem 4.4]{Mitzenmacher-Upfal:2005}).

\begin{theorem}[Chernoff Bounds]
    \label{CH ineq}
        Let $X_1, \dots, X_k$ be a sequence of $k$ independent Bernoulli trials, each succeeding with probability $p$
        so that $E(X_i) = p$.
        Let $S = \sum_{i=1}^k X_i$ be the binomially distributed variable describing the total number of successes.
        Then for $\delta > 0$
        \begin{equation}
            \label{Chernoff bound}
            P \Big(S > (1+\delta)pk \Big) \le \left(\frac{e^{\delta}}{(1+\delta)^{(1+\delta)}}\right)^{pk}.
        \end{equation}
%        Moreover, for $\delta \in (0,1)$,
%        \begin{equation}
%            \label{Chernoff bound 01 delta}
%            P \Big(|S - pk| > \delta pk \Big) \le 2 \exp(-\delta^2 pk).
%        \end{equation}
It follows that for all $a \in (0, p \sqrt{k})$,
\begin{equation}
    \label{Chernoff bound sqrt root}
    P(|S - pk| > a \sqrt{k}) \le 2 \exp(-a^2/p).
\end{equation}
\end{theorem}

\noindent
%{\bf Remark.} We require a slight generalization of \eref{Chernoff bound}
%as follows.
%Suppose that $\hat p\ge p$.  Then for $\delta>0$, 
%\begin{equation}
%            P \Big(S > (1+\delta)\hat pk \Big) \le \left(\frac{e^{\delta}}{(1+\delta)^{(1+\delta)}}\right)^{\hat pk}.
%\label{chernoffgen}
%\end{equation}
%This follows by substituting the quantity $(1+\delta)\hat p/p -1 $ in for
%the ``$\delta$'' appearing in \eref{Chernoff bound} and applying some standard
%inequalities.

%\vspace{0.1in}
The final theorem of this section is as follows.
\begin{theorem}
Let $A$ be an $n\times N$ matrix whose entries are chosen according to
$\Omega$. 
Let $\tilde A$ be defined as follows.  For $(i,j)$
such that $A_{ij}=1$, we define $\tilde A_{ij}=1$.  For entries $(i,j)$
such that $A_{ij}=-p/(1-p)$, we take $\tilde A_{ij} =-n_j/(n-n_j)$, where
$n_j$ is the number of $1$'s in column $j$ of $A$.   
Then there exist $c_1>0$ and $c_2\in(0,1)$ depending on $p$ such that
\begin{equation}
P(\Vert A-\tilde A\Vert_F^2 \le c_1N) \ge 1-(2/3)^{N}-Nc_2^{n}.
\label{eq:bigprob}
\end{equation}
\label{thm:AtildeA}
\end{theorem}

\noindent{\bf Remark 1.} The notation $\Vert A\Vert_F$ denotes the Frobenius
norm of $A$, that is, $\left(\sum_i\sum_j A_{ij}^2\right)^{1/2}$.  
It is well known that $\Vert A \Vert_F\ge \Vert A\Vert$
for any $A$.

\noindent{\bf Remark 2.}
Note that $\tilde A$
is undefined if there is a $j$ such that $n_j=n$.  In this case we
assume that $\Vert A-\tilde A\Vert=\infty$, i.e., the event
considered in \eref{eq:bigprob} fails.

\noindent{\bf Remark 3.}  Observe that the column sums of $A$ are random
variables with mean zero since the mean of the entries is 0.  On the
other hand, the column sums of $\tilde A$ are identically zero
deterministically; this is the
rationale for the choice of  $\tilde A = -n_j/(n-n_j)$.

\begin{proof}
From the definition of $\tilde A$, for column $j$, 
there are exactly $n-n_j$ entries
of $\tilde A$ that differ from those of $A$.  Furthermore,
the difference of these
entries is exactly $(n_j-pn)/((1-p) (n-n_j))$.  Therefore, for each
$j=1,\ldots, N$, the contribution of column $j$ to the square norm
difference $\Vert A-\tilde A\Vert_F^2$ is given by 
$$\Vert A(:,j)- \tilde A(:,j)\Vert_F^2 = \frac{(n_j-pn)^2}{(1-p)^2(n-n_j)}.$$
Recall that the numbers $n_1,\ldots, n_N$ are independent, and each
is the result of $n$ Bernoulli trials done with probability $p$.

We now define $\Psi$ to be the event that at least one $n_j$ is very far
from the mean.  In particular, $\Psi$ is the event that there
exists a $j\in\{1,\ldots,N\}$ such
that $n_j>qn$, where $q=\min(\sqrt{p},2p)$.
Let $\tilde \Psi$ be its complement, and let $\tilde\psi(j)$ be
the indicator of this complement (i.e., $\tilde\psi(j)=1$ if
$n_j\le qn$  else $\tilde\psi(j)=0$).
Let $c$ be a positive scalar depending on $p$ to be determined later.
Observe that 
\begin{eqnarray}
P(\Vert A - \tilde A\Vert_F^2\ge cN) &=&
P(\Vert A - \tilde A\Vert_F^2\ge cN \>\wedge\> \tilde \Psi)
+P(\Vert A - \tilde A\Vert_F^2 \ge cN \>\wedge\> \Psi) \nonumber \\
&\le &
P(\Vert A - \tilde A\Vert_F^2\ge cN \>\wedge\> \tilde \Psi)
+P(\Psi). \label{eq:splitp}
\end{eqnarray}
We now analyze the two terms separately.  For the first term we
use a technique attributed to S.~Bernstein (see Hoeffding
\cite{Hoeffding}).  Let $\phi$ be the indicator function of nonnegative
reals, i.e., $\phi(x)=1$ for $x\ge 0$ while $\phi(x)=0$ for $x<0$.
Then, in general, $P(u\ge 0)\equiv E(\phi(u))$.  Thus,
\begin{eqnarray*}
P(\Vert A - \tilde A\Vert_F^2\ge cN \>\wedge\> \tilde \Psi)
& = &
P(\Vert A - \tilde A\Vert_F^2- cN\ge 0 \>\wedge\> \tilde\psi(n_1)=1
\>\wedge\>\cdots\>\wedge \>\tilde\psi(n_N)=1) \\
&=&
E(\phi(\Vert A - \tilde A\Vert_F^2- cN)\cdot \tilde\psi(n_1)
\cdots\tilde\psi(n_N)).
\end{eqnarray*}
Let $h$ be a positive scalar depending on $p$ to be determined
later.  Observe that for any such $h$ and for all $x\in \R$,
$\phi(x)\le \exp(hx)$.  Thus,
\begin{eqnarray}
P(\Vert A - \tilde A\Vert_F^2\ge cN \>\wedge\> \tilde \Psi)
&\le &
E(\exp(h\Vert A-\tilde A\Vert_F^2-hcN)
\cdot \tilde\psi(n_1)
\cdots\tilde\psi(n_N)) \nonumber \\
& = &
E\left(\exp\left(h\sum_{j=1}^N \left (\Vert A(:,j)-\tilde A(:,j)\Vert_F^2 - c\right)
\right)
\cdot \tilde\psi(n_1)
\cdots\tilde\psi(n_N)\right) \nonumber \\
&=&
E\left(\exp\left(h\sum_{j=1}^N \left( \frac{(n_j-pn)^2}{(1-p)^2(n-n_j)}-c
\right)
\right)
\cdot \tilde\psi(n_1)
\cdots\tilde\psi(n_N)\right)\nonumber \\
&=&
E\left(\prod_{j=1}^N\exp\left( h\left(\frac{(n_j-pn)^2}{(1-p)^2(n-n_j)}-c\right)\right)
\tilde\psi(n_j)\right)\nonumber \\
& = &
\prod_{j=1}^NE\left(\exp\left(h\left(\frac{(n_j-pn)^2}{(1-p)^2(n-n_j)}-c\right)
\right)\tilde\psi(n_j)\right) \label{eq:useindep} \\
&=& f_1\cdots f_N, \label{eq:prodf}
\end{eqnarray}
where
$$f_j =
E\left(\exp\left(h\left(\frac{(n_j-pn)^2}{(1-p)^2(n-n_j)}-c\right)
\right)\tilde\psi(n_j)\right).$$
To obtain \eref{eq:useindep}, we used the independence of the $n_j$'s.
Let us now analyze $f_j$ in isolation.  
\begin{eqnarray*}
f_j 
&=& 
\sum_{i=0}^n
\exp\left(h\left(\frac{(i-pn)^2}{(1-p)^2(n-i)}-c\right)
\right)\tilde\psi(n_j)P(n_j=i)
\\
&=& 
\sum_{i=0}^{\lfloor qn\rfloor}
\exp\left(h\left(\frac{(i-pn)^2}{(1-p)^2(n-i)}-c\right)
\right)P(n_j=i)
\\
&\le& 
\sum_{i=0}^{\lfloor qn\rfloor}
\exp\left(h\left(\frac{(i-pn)^2}{(1-p)^2(n-\sqrt{p}n)}-c\right)
\right)P(n_j=i).
\end{eqnarray*}
To derive the last line, we used the fact that $i\le\sqrt{p}n$ since
$i\le qn$.
Now let us reorganize this summation by considering first 
$i$ such that $|i-pn|< \sqrt{n}$, and next
$i$ such that $|i-pn|\in[\sqrt{n},2\sqrt{n})$, etc.
Notice that, since $i\le qn\le 2pn$, we need consider
intervals only until $|i-pn|$ reaches $pn$.
\begin{eqnarray*}
f_j & \le &
\sum_{k=0}^{\lfloor p\sqrt{n}\rfloor}\sum_{i:|i-pn|\in [k\sqrt{n},(k+1)\sqrt{n})}
\exp\left(h\left(\frac{(i-pn)^2}{(1-p)^2(n-\sqrt{p}n)}-c\right)
\right)P(n_j=i) 
\\
&\le &
\sum_{k=0}^{\lfloor p\sqrt{n}\rfloor}\sum_{i:|i-pn|\in [k\sqrt{n},(k+1)\sqrt{n})}
\exp\left(h\left(\frac{(k+1)^2n}{(1-p)^2(n-\sqrt{p}n)}-c\right)
\right)P(n_j=i)  \\
&=&
\sum_{k=0}^{\lfloor p\sqrt{n}\rfloor}\sum_{i:|i-pn|\in [k\sqrt{n},(k+1)\sqrt{n})}
\exp\left(h\left(\frac{(k+1)^2}{(1-p)^2(1-\sqrt{p})}-c\right)
\right)P(n_j=i)  \\
&=&
\sum_{k=0}^{\lfloor p\sqrt{n}\rfloor}
\exp\left(h\left(\frac{(k+1)^2}{(1-p)^2(1-\sqrt{p})}-c\right)
\right)\sum_{i:|i-pn|\in [k\sqrt{n},(k+1)\sqrt{n})} P(n_j=i)  \\
&\le &
2\sum_{k=0}^{\lfloor p\sqrt{n}\rfloor}
\exp\left(h\left(\frac{(k+1)^2}{(1-p)^2(1-\sqrt{p})}-c\right)
\right)\exp(-k^2/p),
\end{eqnarray*}
where, for the last line, we have applied \eref{Chernoff bound sqrt root}.
The theorem is valid since $k\le p\sqrt{n}$.

Continuing this derivation and overestimating the finite
sum with an infinite sum,
\begin{eqnarray*}
f_j &\le& 
2\exp(-hc)\cdot\sum_{k=0}^{\infty}
\exp\left(\frac{h(k+1)^2}{(1-p)^2(1-\sqrt{p})}-k^2/p\right) \\
&=&
2\exp\left(\frac{h}{(1-p)^2(1-\sqrt{p})}-hc \right) \\
& & \qquad\mbox{}
+
2\exp(-hc)\cdot\sum_{k=1}^{\infty}
\exp\left[\frac{h(k+1)^2}{(1-p)^2(1-\sqrt{p})}-k^2/p\right]
.
\end{eqnarray*}
Choose $h$ so that $h/((1-p)^2)(1-\sqrt{p})<1/(8p)$, i.e.,
$h<(1-p)^2(1-\sqrt{p})/(8p)$.  Then the second term in the square-bracket
expression
at least twice the first term for all $k\ge 1$, hence
\begin{equation}
f_j
\le 
2\exp\left(\frac{h}{(1-p)^2(1-\sqrt{p})}-hc\right) +
2\exp(-hc)\cdot\sum_{k=1}^{\infty}
\exp\left(-k^2/(2p)\right).
\label{eq:fjfinal}
\end{equation}
Observe that $\sum_{k=1}^\infty \exp(-k^2/(2p))$ is dominated by
a geometric series and hence is a finite number depending on $p$.
Thus, once $h$ is selected, it is possible to choose $c$ sufficiently large
so that each of the two terms in \eref{eq:fjfinal} is at most $1/3$.
Thus, with 
appropriate choices of $h$ and $c$, we conclude that $f_j\le 2/3$.
Thus, substituting this into \eref{eq:prodf} shows that
\begin{equation}
P(\Vert A - \tilde A\Vert_F^2\ge cN \>\wedge\> \tilde \Psi)\le (2/3)^{N}.
\label{eq:p1stterm}
\end{equation}
We now turn to the second term in \eref{eq:splitp}.  For a particular
$j$, the probability that $n_j> qn$ is bounded using \eref{Chernoff bound}
by $v_p^{n}$ where $v_p=(e^\delta/(1+\delta)^{(1+\delta)})^p$, where
$\delta=q/p-1$, i.e., $\delta=\min(p,\sqrt{p}-p)$.  Then the union bound
asserts that the probability that any $j$ satisfies $n_j>qn$ is at most
$Nv_p^n$.  Thus,
$$P(\Vert A-\tilde A\Vert_F^2\ge cN)\le (2/3)^N+Nv_p^n.$$
This concludes the proof.
\end{proof}

\section{Maximum Clique}
\label{sec:maxclique}

Let $G = (V,E)$ be a simple graph.
The \textbf{maximum clique problem} focuses on finding the largest clique of graph $G$, i.e., the largest complete subgraph of $G$.
For any clique $K$ of $G$, the adjacency matrix of the graph $K'$ obtained by taking the union of $K$ and the set of loops for each $v \in V(K)$ is a rank-one matrix
with 1's in the entries indexed by $V(K)\times V(K)$ and 0's everywhere else.
Therefore, a clique $K$ of $G$ containing $n$ vertices can be found by solving the rank minimization problem
\begin{align}
    \min \;  & \rank(X)  \notag \\
    \st  \;  & \sum_{i \in V} \sum_{j \in V} X_{ij} \ge n^2, \label{clique sum} \\
             & X_{ij} = 0 \;\; \mbox{if } (i,j) \notin E \mbox{ and } i \neq j,  \label{clique missing edges} \\
             & X \in [0, 1]^{V \times V}	\label{clique binomial}.
\end{align}
Unfortunately, this rank minimization problem is also NP-hard. 
We consider the relaxation obtained by replacing the objective function with the nuclear norm, the sum of the singular values of the matrix:
$
	\|X\|_* = \sigma_1(X) +\cdots + \sigma_N(X).
$

Underestimating $\rank(X)$ with $\|X\|_*$, we obtain the following convex optimization problem:
\begin{align}
	\begin{array}{rl}
    		\min 	& \|X\|_*  \\
    		\st  	& \sum_{i \in V} \sum_{j \in V} X_{ij} \ge n^2, \\
            		& X_{ij} = 0 \;\; \mbox{if } (i,j) \notin E \mbox{ and } i \neq j.  
	\end{array}            
            \label{clique relaxation}
\end{align}
Notice that the relaxation has dropped the constraint 
$X_{ij}\le 1$ that was present in the original formulation.  This constraint
turns out to be superfluous (and, in fact, unhelpful---see the remark
following \eref{clique KKT}) for our approach.
Using the Karush-Kuhn-Tucker conditions, we derive conditions for which the adjacency matrix of a graph comprising a clique of $G$ of size $n$
together with $n$ loops for each vertex in the clique is optimal for this convex relaxation.

\subsection{Optimality Conditions}
\label{sec:maxcliqueopt}

In this section, we prove a theorem that gives sufficient conditions
for optimality and uniqueness of a solution to \eref{clique relaxation}.
These conditions involve multipliers
$\lambda_{ij}$ and $\mu$ and a matrix $W$.  In subsequent 
subsections we explain
how to select
$\lambda_{ij}$, $\mu$ and $W$ based on the underlying graph to
satisfy the conditions.

Recall that if $f:\R^n\rightarrow\R$ is
a convex function, then a {\em subgradient} of $f$ at a point $\x$
is defined to be a vector $\g\in\R^n$ such that for all $\y\in\R^n$,
$f(\y)-f(\x)\ge \g^T(\y-\x)$.  It is a well-known theorem that for a convex
$f$ and for every $\x\in\R^n$, 
the set of subgradients forms a nonempty closed convex
set.  This set of subgradients, called the {\em subdifferential},
is denoted as $\partial f(\x)$.

In this section we consider the following generalization of 
\eref{clique relaxation} because it will also arise in our discussion
of biclique below:
\begin{equation}
\begin{array}{rl}
\min & \Vert X\Vert_* \\
\mbox{s.t.} &  \sum_{i=1}^M\sum_{j=1}^N X_{i,j}\ge mn, \\
& X_{i,j}=0 \mbox{ for $(i,j)\in \tilde E$}
\end{array}
\label{eq:relax1}
\end{equation}
Here, $X\in\R^{M\times N}$,
$E$ is a subset of $\{1,\ldots,M\}\times\{1,\ldots, N\}$, and
the complement of $E$ is denoted $\tilde E$.

The following lemma characterizes the subdifferential of $\|\cdot \|_*$ (see \cite[Equation 3.4]{Candes-Recht:2008} and also \cite{Watson:1992}).

\begin{lemma}
    \label{lemma: subdifferential of nuclear norm}
    Suppose $A \in \R^{m\times n}$ has rank $r$ with
    singular value decomposition $A = \sum_{k=1}^r \sigma_k \u_k \v_k^T$.
    Then
    $\phi$ is a subgradient of $\|\cdot\|_*$ at $A$ if and only if $\phi$ is of the form
    \[
        \phi = \sum_{k=1}^r \u_k \v_k^T + W
    \]
    where $W$ satisfies $\|W\| \le 1$ such that the column space of $W$ is orthogonal to $\u_k$ and the row space of $W$
            is orthogonal to $\v_k$ for all $k=1,2,\ldots,r$.
\end{lemma}

Let $I$ be a subset of $\{1,\ldots,N\}$.  We say that $\u\in\R^N$ is the
{\em characteristic vector} of $I$ if $u_i=1$ for $i\in I$ while
$u_i=0$ for $i\in\{1,\ldots,N\}-I$.  
%Given an $\x\in\R^N$, the notation
%$\x_I$ means the subvector of $\x$ indexed by $I$.

Let $U^*$ be a subset of $\{1,\ldots,M\}$ and $V^*$ a subset of
$\{1,\ldots,N\}$, and let $\bar\u$, $\bar\v$ be their characteristic
vectors respectively.  Suppose $|U^*|=m$ and $|V^*|=n$ with $m>0$, $n>0$.
Let $X^*=\bar \u\bar\v^T$, an $M\times N$ matrix.
Clearly $X^*$ has rank 1.  Note that
Lemma~\ref{lemma: subdifferential of nuclear norm} implies that
\begin{equation}
    \partial \|\cdot\|_*(X^*) = \{ {\bar{\u}\bar{\v}^T}/{\sqrt{mn}} + W: W \bar{\v} = \bz, \; W^T \bar{\u} =\bz, \; \|W\| \le 1\}.
\label{eq:rank1subgrad}
\end{equation}

This leads to the main theorem for this section.

\begin{theorem}
Let $U^*$ be a subset of $\{1,\ldots, M\}$ of cardinality $m$, and let
$V^*$ be a subset of $\{1,\ldots,N\}$ of cardinality
$n$.  Let $\bar\u$ and $\bar\v$ be the characteristic vectors of $U^*$, $V^*$
respectively.
Let $X^*=\bar\u\bar\v^T$.  Suppose $X^*$ is feasible for
\eref{eq:relax1}.  Suppose also that there
exist $W\in\R^{M\times N}$, $\lambda \in\R^{M\times N}$ and
$\mu\in\R_+$ such that
$W \bar{\v} = \bz $, $\bar{\u}^T W =\bz$, $\|W \| \le 1$ and
\begin{align}
    \label{clique KKT0}
    \frac{\bar{\u}\bar{\v}^T}{\sqrt{mn}} + W &= \mu \e \e^T
                                        + \sum_{\substack{(i,j) \in \tilde E }} \lambda_{ij} \e_i \e_j^T.
\end{align}
Here, $\e$ denotes the vector of all $1$'s while $\e_i$ denotes
the $i$th column of the identity matrix  (either in $\R^M$ or
$\R^N$).
Then $X^*$ is an optimal solution to \eref{eq:relax1}.  Moreover,
for any $I\subset\{1,\ldots,M\}$ and $J\subset\{1,\ldots,N\}$
such that $I\times J\subset E$, $|I|\cdot|J|
\le mn$.  

Furthermore, if $\Vert W\Vert<1$
and $\mu > 0$,
then $X^*$ is the unique optimizer of \eref{eq:relax1} (and hence
will be found if a solver is applied to \eref{eq:relax1}).
\label{thm:optimality1}
\end{theorem}

\begin{proof}
The fact that $X^*$ is optimal is a straightforward application
of the well-known KKT conditions. Nonetheless,
we now explicitly prove optimality
because the inequalities in the proof are useful for the
uniqueness proof below.

Suppose $X$ is another matrix feasible for \eref{eq:relax1}.  
We wish to show that $\Vert X\Vert_*\ge \Vert X^*\Vert_*$.  To prove
this, we use the definition of subgradient followed by 
\eref{clique KKT0}.  The notation
$A\bullet B$ is used
to denote the elementwise inner product of two matrices $A,B$.
\begin{eqnarray}
\Vert X\Vert_*-\Vert X^*\Vert_* &\ge &
(\bar{\u}\bar{\v}^T/\sqrt{mn} + W)\bullet (X-X^*)  \label{eq:subgrad}\\
&=& \mu (\e \e^T)\bullet (X-X^*)
     + \sum_{\substack{(i,j) \in \tilde E }} \lambda_{ij} 
          (\e_i \e_j^T)\bullet (X-X^*) \label{eq:kkt1} \\
&=&\mu \left((\e\e^T)\bullet X -mn \right) \label{eq:kkt1simp} \\
&\ge& 0. \label{eq:gez}
\end{eqnarray}
Equation \eref{eq:subgrad} follows by the definition of subgradient and
\eref{eq:rank1subgrad}; \eref{eq:kkt1} follows from \eref{clique KKT0};
and \eref{eq:kkt1simp} follows from the fact that $(\e\e^T)\bullet X^*=mn$
by definition of $X^*$ and 
$(\e_i\e_j^T)\bullet X=(\e_i\e_j^T)\bullet X^*=0$ for $(i,j)\in \tilde E$
by feasibility. Finally, \eref{eq:gez} follows since $\mu\ge 0$
and $(\e\e^T)\bullet X\ge mn$ by feasibility.
This proves that $X^*$ is an optimal solution to \eref{eq:relax1}.

Now consider $(I,J)$ such that $I\times J\subset E$.  Then
$X'=\bar\u'(\bar \v')^T\cdot mn/(|I|\cdot|J|)$, 
where $\bar\u'$ is the characteristic vector of $I$ and
$\bar\v'$ is the characteristic vector of $J$, is also a feasible
solution to \eref{eq:relax1}.   Recall that for a matrix of the form $\u\v^T$, the
unique nonzero singular value (and hence the nuclear norm) equals
$\Vert\u\Vert\cdot\Vert\v\Vert$.  Thus, $\Vert X'\Vert_*=mn/(|I|\cdot|J|)^{1/2}$
and $\Vert X^*\Vert_*=\sqrt{mn}$.  
Since $X^*$ is optimal, $\Vert X'\Vert_*\ge
\Vert X^*\Vert$, i.e., $\sqrt{mn}\le mn/(|I|\cdot|J|)^{1/2}$.  Simplifying
yields $|I|\cdot|J|\le mn$.

Now finally we turn to the uniqueness of $X^*$, which is the
most complicated part of the proof.  This argument requires
a preliminary claim.  Let $S_1$ denote the subspace of $M\times N$ matrices
$Z_1$ such that $\bar\u^TZ_1=\bz$ and $Z_1\bar\v=\bz$.  Let $S_2$ denote the
subspace of $M\times N$ matrices that can be written in the form
$\x\bar \v^T$, where $\x\in\R^M$ has all zeros in positions indexed by
$U^*$.  Let $S_3$ denote the subspace of $M\times N$ matrices that can
be written in the form $\bar \u\y^T$, where $\y\in\R^N$ has all zeros in
positions indexed by $V^*$.  Let $S_4$ denote the subspace of all 
$M\times N$ matrices that can be written in the form $\bar\u\y^T+\x\bar\v^T$,
where $\x$ has nonzeros only in positions indexed by $U^*$, $\y$ has
nonzeros only in positions indexed by $V^*$, and 
the sum of entries of $\bar\u\y^T+\x\bar\v^T$ is zero.
Finally, let $S_5$ be the subspace of $M\times N$ matrices of the form
$\alpha\bar\u\bar\v^T$, where $\alpha$ is a scalar.

The preliminary claim is that $S_1,\ldots,S_5$ are mutually orthogonal and that
$S_1\oplus\cdots\oplus S_5=\R^{M\times N}$.  To check orthogonality, we
proceed case by case.  For example, if $Z_1\in S_1$ and $Z_2\in S_2$, then
$Z_2=\x\bar\v^T$ so $Z_1\bullet Z_2=Z_1\bullet (\x\bar\v^T)=\x^TZ_1\bar\v=0$
since $Z_1\bar\v=\bz$.  The identity $Z\bullet(\x\y^T)=\x^TZ\y$ similarly shows
that $Z_1$ is orthogonal to all of $S_2,\ldots,S_5$.
Next, observe that $Z_2\in S_2$ has nonzero entries only in positions
indexed by $U^*\times \tilde V^*$, where $\tilde V^*$ denotes
$\{1,\ldots,N\}-V^*$.  Similarly, $Z_3\in S_3$ has nonzero entries
only in positions indexed by $\tilde U^*\times V^*$, and $Z_4\in S_4$
and $Z_5\in S_5$ have nonzero entries only in positions indexed by 
$U^*\times V^*$.  Thus, the nonzero entries of $S_2$, $S_3$ and $S_4\oplus S_5$
are disjoint, and hence these spaces are mutually orthogonal.  The only
remaining case is to show that $S_4$ and $S_5$ are orthogonal; this
follows because a matrix in $S_5$ is a multiple of the all $1$'s
matrix in positions indexed by $U^*\times V^*$, 
while the entries of a matrix in $S_4$, also only in positions indexed
by $U^*\times V^*$, sum to $0$.

Now we must show that $S_1\oplus\cdots\oplus S_5=\R^{M\times N}$.  Select
a $Z\in\R^{M\times N}$.  We first split off an $S_5$ component:
let $\alpha=\bar\u^TZ\bar\v/((\bar\u^T\bar\u)(\bar\v^T\bar\v))$
and define $Z_5=\alpha\bar\u\bar\v^T$.  Then $Z_5\in S_5$.  Let
$\dot Z=Z-Z_5$.  One checks from the definition of $\alpha$ that
$\bar\u^T\dot Z\bar\v=0$.  It remains to write $\dot Z$ as a matrix
in $S_1\oplus\cdots\oplus S_4$.

Next we split off an $S_1$ component.
Let $\x=\dot Z\bar\v/\bar\v^T\bar\v$ and $\y=\dot Z^T\bar\u/\bar\u^T\bar\u$.
Observe that $\bar\u^T\x=\bar\u^T\dot Z\bar\v/\bar\v^T\bar\v=0$.  Similarly,
$\bar\v^T\y=0$.  Let $\ddot Z=\x\bar\v^T+\bar\u\y^T$ and $Z_1=\dot Z-\ddot Z$.
Then 
\begin{eqnarray*}
Z_1\bar\v & = & 
\dot Z\bar\v - \ddot Z\bar\v \\
& = &\dot Z\bar\v - \x\bar\v^T\bar\v
-\bar\u\y^T\bar\v \\
& = &
\dot Z\bar\v - \x\bar\v^T\bar\v \\
& = &
\bz,
\end{eqnarray*}
where the third line follows because $\bar\v^T\y=0$ and the fourth by
definition of $\x$.  Similarly, $Z_1^T\bar\u=\bz$.  Thus, $Z_1\in S_1$.

It remains to split $\ddot Z$ among $S_2$, $S_3$ and $S_4$.  
Write $\x=\x_1+\x_2$, where $\x_1$ is nonzero only in entries indexed
by $U^*$ while $\x_2$ is nonzero only in entries indexed by $\tilde U^*$.
 Similarly,
split $\y=\y_1+\y_2$ using $V^*$ and $\tilde V^*$.
Then $\ddot Z=\x_1\bar\v^T+\x_2\bar\v^T+\bar\u\y_1^T
+\bar\u\y_2^T$.  Then $\x_2\bar\v^T\in S_2$ and $\bar\u\y_2^T\in S_3$,
so define $Z_2=\x_2\bar\v^T$ and $Z_3=\bar\u\y_2^T$.
Finally, we must consider the remaining term 
$Z_4=\ddot Z -\x_2\bar\v^T-\bar\u\y_2^T= \x_1\bar\v^T+\bar\u\y_1^T$.
This has the form required for membership in $S_4$, but it remains to
verify that the sum of entries of $Z_4$ add to zero.  This is shown
as follows:
\begin{eqnarray*}
Z_4\bullet (\e\e^T)&=& Z_4\bullet(\bar\u \bar\v^T) \\
&=& \bar\u^TZ_4\bar\v \\
& = & (\bar\u^T\x_1)(\bar\v^T\bar\v) + (\bar\u^T\bar\u)(\y_1^T\bar\v) \\
& = &(\bar\u^T\x)(\bar\v^T\bar\v) + (\bar\u^T\bar\u)(\y^T\bar\v) \\
&=& 0+0.
\end{eqnarray*}
The second line follows because $Z_4$ is all zeros outside entries indexed
by $U^*\times V^*$.  The fourth line follows because $\bar\u$ is zero
outside $U^*$ and similarly for $\bar \v$.  The last line follows from 
equalities derived in the previous paragraph.

This concludes the proof of the claim that $S_1,\cdots, S_5$ split 
$\R^{M\times N}$ into mutually orthogonal subspaces.

Now we prove the uniqueness of $X^*$ under the assumption that 
$\mu>0$ and $\Vert W\Vert < 1$.  Let $X$ be a feasible solution
different from $X^*$.  Write $X-X^*=Z_1+\cdots+Z_5$, where
$Z_1,\ldots,Z_5$ lie in $S_1,\ldots,S_5$ respectively.  Now we
consider several cases.

The first case is that $Z_1\ne 0$.  Then since $\Vert W\Vert <1$
and $Z_1\bar\v=\bz$, $Z_1^T\bar\u=\bz$, it follows 
from Lemma~\ref{lemma: subdifferential of nuclear norm}
that $W+\eps Z_1$ lies
in $\partial\Vert\cdot\Vert_*(X^*)$ for $\eps>0$ sufficiently small.
This means that `$W$' appearing in \eref{eq:subgrad} above may be
replaced by $W+\eps Z_1$ without harming the validity of the inequality.
This adds the term $\eps Z_1\bullet(X-X^*)$ to the right-hand sides of the 
inequalities following \eref{eq:subgrad}.  Observe that
$Z_1\bullet(X-X^*)=Z_1\bullet(Z_1+\cdots+Z_5)=Z_1\bullet Z_1>0$.
Thus, a positive quantity is added to all these right-hand sides,
so we conclude $\Vert X\Vert_*-\Vert X^*\Vert_*>0$.

For the remaining cases, we assume $Z_1=0$.  We claim that 
$Z_2=Z_3=0$ as well.  For example, suppose $Z_2=\x\bar\v^T$.
Recall that $Z_2$ is nonzero only for entries indexed by
$\tilde U^*\times V^*$ (and in particular, $\x$ must be zero
on $U^*$).
Since all of $Z_3$, $Z_4$ and $Z_5$ are zero in
$\tilde U^*\times V^*$, $Z_2(i,j)=X_{ij}-X^*_{ij}$ for
$(i,j)\in \tilde U^*\times V^*$.
Select an $i\in\tilde U^*$; we claim that there exists a $j\in V^*$ such that
$(i,j)\notin E$.  If not, then $(U^*\cup\{i\})\times V^*$ would define
a solution to \eref{eq:relax1} with greater cardinality (and hence
lower objective value) than $U^*\times V^*$, but we have already proven that
$U^*\times V^*$ defines the optimal solution.
Thus, there is a constraint in \eref{eq:relax1} of the form
$X_{i,j}=0$ that must be satisfied by both $X$ and $X^*$.  This means
that the $(i,j)$ entry of $Z_2$ is zero.  On the other hand, this
entry is $x_i\bar v_j=x_i$.  Thus, we conclude $x_i=0$.  Therefore,
$\x=\bz$ so $Z_2$ vanishes.  The same argument shows $Z_3$ vanishes.

The last case is thus that $Z_1$, $Z_2$ and $Z_3$ are all zero, so at
least one of $Z_4$ or $Z_5$ must be nonzero.  Since the sum of entries
of $Z_4$ is zero and $X$ is feasible
(and, in particular, feasible for the constraint
$X\bullet(\e\e^T)\ge mn$), it follows that the sum of entries
of $Z_5$ must be nonnegative, i.e., $Z_5=\alpha \bar\u\bar\v^T$ with
$\alpha \ge 0$.  If $\alpha >0$ then we are finished with the proof:
the assumption $\mu>0$ and $\alpha>0$ imply that both factors
in \eref{eq:kkt1simp} are positive, hence $\Vert X\Vert_*-\Vert X^*\Vert_*>0$.

Thus, we may assume that $Z_5=0$ so $Z_4\ne 0$.
Recall that $Z_4$ is nonzero only in positions
indexed by $U^*\times V^*$.
We can now draw the following conclusions about the singular values
of $X$ versus those of $X^*$.  Recall that the rank of $X^*$ is one, and
its sole nonzero singular value is $\sqrt{mn}$ and hence
$\Vert X^*\Vert_F=\Vert X^*\Vert=\Vert X^*\Vert_*=\sqrt{mn}$.
Observe that the sum
of entries of $X$, namely, $\bar\u^T X\bar\v$, is also $mn$.  
But $\bar \u^T X\bar\v\le \Vert\bar\u\Vert\cdot\Vert X\Vert\cdot\Vert\bar\v\Vert
=\Vert X\Vert \sqrt{mn}$.  Thus, $\Vert X\Vert\ge \sqrt{mn}$,
i.e., $\sigma_1(X)\ge \sigma_1(X^*)$, where $\sigma_k(A)$ is notation for
the $k$th singular value of matrix $A$.

Next, note that $\Vert X\Vert_F>\Vert X^*\Vert_F$ for the following
reason.  Recall that the Frobenius norm is equivalent to the Euclidean
vector norm applied to the matrix when regarded as a vector.  Furthermore,
when regarded as a vector, $X$ is the sum of two orthogonal components,
namely $X^*$ and $Z_4$.  Therefore, by the Pythagorean theorem,
$\Vert X\Vert_F=\left(\Vert X^*\Vert_F^2+\Vert Z_4\Vert^2\right)^{1/2}$.
Since $Z_4\ne 0$, $\Vert X\Vert_F>\Vert X^*\Vert_F$.

Thus, we know that $\sigma_1(X)\ge \sigma_1(X^*)$ and
that $\sigma_1(X)^2+\sigma_2(X)^2>\sigma_1(X^*)^2$.
These two inequalities imply that $\sigma_1(X)+\sigma_2(X)>\sigma_1(X^*)$,
and therefore $\Vert X\Vert_*>\Vert X^*\Vert_*$.

Thus, we have shown that in all cases, if $\Vert W\Vert<1$,
$\mu>0$ and $X$ is a feasible point distinct from $X^*$, then
$\Vert X\Vert_*>\Vert X^*\Vert_*$.  This proves that $X^*$ is the
unique optimizer.
\end{proof}

This theorem immediately specializes to the following theorem if
we take the case that $G$ is an $N$-node undirected
graph, that $M=N$, $m=n$, and 
$E=E(G)\cup\{(i,i):i\in V(G)\}$.

\begin{theorem}
Let $V^*$ be the nodes of an $n$-node clique contained in
an $N$-node undirected
graph $G=(V,E)$.
Let $\bar\v\in\R^V$ be the characteristic vector of $V^*$.
Let $X^*=\bar\v\bar\v^T$.  (Clearly $X^*$ is feasible for
\eref{clique relaxation}).  Suppose also that there
exist $W\in\R^{V\times V}$, $\lambda \in\R^{V\times V}$ and
$\mu\in\R_+$ such that
$W \bar{\v} = \bz $, $\bar{\v}^T W = \bz$, $\|W \| \le 1$ and
\begin{align}
    \label{clique KKT}
    \frac{\bar{\v}\bar{\v}^T}{n} + W &= \mu \e \e^T
                                        + \sum_{\substack{(i,j) \in \tilde E }} \lambda_{ij} \e_i \e_j^T.
\end{align}
Then $X^*$ is an optimal solution to \eref{clique relaxation}.  
Moreover, $V^*$ is a maximum clique of $G$.
Furthermore, if $\Vert W\Vert<1$
and $\mu > 0$,
then $X^*$ is the unique optimizer of \eref{clique relaxation},
and $V^*$ is the unique maximum clique of $G$.
\label{thm:optimalityclique}
\end{theorem}

\noindent
{\bf Remark:} It may appear that we need to know the value of $n$
prior to applying the theorem since $n$ is present in the statement of
\eref{clique relaxation}.  In fact, this is not the case: we observe
that the factor $n^2$ appearing in \eref{clique relaxation} is the
sole inhomogeneity in the problem.  This means that we obtain
the same solution, rescaled in the appropriate way, if we replace $n^2$
by $1$ in \eref{clique relaxation}.  Thus, $n$ does not need to be
known in advance to apply this theorem.

\vspace{0.1in}
For the next two subsections, we consider two scenarios for constructing
 $G$
and try to find $X^*$, $W$ and values for the
multipliers to satisfy the conditions of the previous theorem.
For both subsections, we use the following choices.
We take $\mu = 1/n$ where $n=|V^*|$.
We define $W$ and $\lambda$ by considering the following cases:

\begin{itemize}
    \item[($\omega_1$)]
        If $(i,j) \in V^* \times V^*$, we choose $W_{ij} =0$ and $\lambda_{ij} = 0$.
        In this case, the entries on other side of (\ref{clique KKT}) corresponding to this case become $1/n + 0 = 1/n + 0$.
    \item[($\omega_2$)]
        If $(i,j) \in E - (V^*\times V^*)$ such that $i \neq j$, then we choose $W_{ij} = 1/n$ and $\lambda_{ij} = 0$.
        Then the two sides of (\ref{clique KKT}) become $0 + 1/n = 1/n + 0.$
    \item[($\omega_3$)]
        If $i \notin V^*,$ we set $W_{ii} = 1/n$. Again the two sides of (\ref{clique KKT}) become
        $ 0 + 1/n = 1/n + 0$.
    \item[($\omega_4$)]
        If $(i,j) \notin E$, $i \notin V^*,$ $j\notin V^*$, then we choose $W_{ij} = -\gamma /n$ and $\lambda_{ij} = -(1+ \gamma)/n$ for
        some constant $\gamma \in \R$.
        The two sides of (\ref{clique KKT}) become $0 - \gamma/n = 1/n - (1 + \gamma)/n$.
        The value of $\gamma$ is specified below.
    \item[($\omega_5$)]
        If $(i,j) \notin E$, $i \in V^*$, $j\notin V^*$, then we choose
        \[
            W_{ij} = - \frac{p_j}{n(n-p_j)}, \;\;\; \lambda_{ij} = -\frac{1}{n} - \frac{p_j}{n(n-p_j)}
        \]
        where $p_j$ is equal to the number of edges in $E$ from $j$ to $V^*$.
    \item[($\omega_6$)]
        If $(i,j) \notin E,$ $i\notin V^*$, $j \in V^*$ then choose $W_{ij}$, $\lambda_{ij}$ symmetrically with the previous case.
\end{itemize}

First, observe that $W \bar{\v} = \bz$.
Indeed, for entries $i \in V^*$, $W(i,:) \bar{\v} = 0$ since $W(i, V^*) = 0$ for such entries.
For entries $i \in V-V^*$,
\[
    W(i,:) \bar{\v} = p_i \frac{1}{n} - (n- p_i) \frac{p_i}{n(n-p_i)} = 0
\]
by our special choice of $W(i,j)$ in cases 5 and 6.

It remains to determine which graphs $G$ yield $W$ as defined by ($\omega_1$)--($\omega_6$) such that $\|W\| < 1$.
We present two different analyses.

\subsection{The Adversarial Case}
\label{sec:maxcliqueadver}

Suppose that the edge set of the graph $G = (V,E)$ is generated as follows. We first add a clique $K_{V^*}$ with vertex set $V^*$ of size $n$.
Then, an adversary is allowed to add a number of the remaining $|V|(|V|-1)/2 - n(n-1)/2$ potential edges to the graph.
We will show that, under certain conditions, our adversary can add up to $O(n^2)$ edges to the graph and $K_V^*$ will still be the unique maximum clique of $G$.

We first introduce the following notation. Let $W^D \in \R^{V\times V}$ denote the matrix with diagonal entries equal to the
diagonal entries of $W$ and all other entries equal to 0. Let $W^{ND}$ be the matrix whose nondiagonal entries are equal to the
corresponding nondiagonal entries of $W$ and whose diagonal entries are equal to 0. So $W = W^D + W^{ND}$.

Now suppose $G = (V,E)$ contains a clique $K_{V^*}$ of size $n$ with vertices indexed by $V^* \in \R^V$. Moreover, suppose that $G$
contains at most $r$ edges not in $K_{V^*}$ and each vertex in $V-V^*$ is adjacent to at most $\delta n$ vertices in $V^*$ for some $\delta \in (0,1)$.
Consider $W$ as defined by ($\omega_1$)--($\omega_6$) with $\gamma = 0$.
By the triangle inequality,
\[
    \|W\|^2 \le (\|W^D\| + \|W^{ND}\|)^2 \le 2(\|W^D\|^2 + \|W^{ND}\|^2) = 2(1/n^2 + \|W^{ND}\|^2)
\]
since $\|W^D\| = 1/n$.
Applying the bound $\|W\| \le \|W\|_F$, it suffices to determine which values of $r$ yield
\[
    \|W^{ND}\|_F^2 = 2\|W(V^*, V-V^*)\|^2_F + \|W^{ND}(V-V^*, V-V^*)\|^2_F < (n^2 - 2)/(2n^2)
\]
since, by the symmetry of $W$,
\[
	W^{ND}(V^*, V-V^*) = W(V^*, V-V^*) = W(V - V^*, V^*).
\]
The diagonal entries of $W^{ND}(V-V^*, V-V^*)$ are equal to $0$ and at most $2r$ of the remaining entries are equal to $1/n$.
Therefore,
\[
    \|W^{ND}(V-V^*,V-V^*)\|^2_F \le {2r}/{n^2}.
\]
Moreover, since $n - p_j \ge (1-\delta) n$,
\begin{align*}
    \|W(V^*,V-V^*)\|^2_F 	&= \sum_{j\in V-V^*}\left(  p_j \cdot \frac{1}{n^2} + (n - p_j) \cdot \frac{p_j^2}{(n - p_j)^2 n^2} \right)     \\
    						&= \sum_{j \in V-V^*} \left( \frac{p_j}{n^2} + \frac{p_j^2}{(n - p_j) n^2} \right) \\
						&\le \sum_{j \in V-V^*} \left( \frac{p_j}{n^2} + \frac{ \delta n p_j }{(1-\delta) n^3} \right) \\
						&= \left( \frac{1}{1-\delta} \right) \sum_{j \in V-V^*} \frac{p_j}{n^2}  \\
						&\le \left( \frac{1}{1-\delta} \right) \frac{r}{n^2}.
\end{align*}
Thus, the optimality and uniqueness conditions given by Theorem~\ref{thm:optimality1} are satisfied by $X^*$ if
\[
   \left( 1 + \frac{1}{1-\delta} \right) r < (n^2-2)/4.
\]
Equivalently,
\[
	r < \frac{1 - \delta}{4 (2 - \delta)}(n^2 - 2).
\]
Therefore, $G$ can contain up to $O(n^2)$ edges other than those in $V^* \times V^*$, and yet $V^*$ will remain the unique maximum clique of $G$.

Note that these bounds are the best possible up to the constant factors.  In particular,
if the adversary were able to insert $(n+1)(n+2)/2$ edges, then a new clique could be created
larger than the planted clique.  Thus, the adversary must be limited to 
${\rm const}\cdot n^2$ edges for ${\rm const}<1/2$.  Similarly, if the adversary could
join a nonclique vertex to $n$ clique vertices, then the adversary would have enlarged
the clique.  Thus, the restriction that a nonclique vertex is adjacent to at most 
${\rm const}\cdot n$ clique vertices is the best possible.

\subsection{The Randomized Case}
\label{sec:maxcliquerand}

Let $V$ be a set of vertices with $|V| = N$ and consider a subset $V^* \subseteq V$ such that $|V^*| = n$.
We construct the edge set $E$ of the graph $G = (V,E)$ as follows:
\begin{enumerate}
    \item[($\Gamma_1$)]
        For all $(i,j) \in V^*\times V^*$, $(i,j) \in E$.
    \item[($\Gamma_2$)]
        Each of the remaining $N(N - 1)/2 - n(n-1)/2$ possible edges is added to $E$ independently at random with probability $p \in [0,1)$ .
\end{enumerate}
Notice that, by our construction of $E$, $G$ contains a clique of size $n$ with vertices indexed by $V^*$.
We wish to determine which $n$, $N$ yield $G$ as constructed by ($\Gamma_1$) and ($\Gamma_2$) such that with
high probability $X^* = \bar{\v} \bar{\v}^T$ is optimal  for the convex relaxation of the clique problem given by (\ref{clique relaxation}).
The following theorem states the desired result.

\begin{theorem}
    \label{main random case}
    There exists an $\alpha > 0$ depending on $p$ 
such that for all $G$ constructed via $(\Gamma_1)$, $(\Gamma_2)$ with $n \ge \alpha \sqrt{N}$,
    the clique defined by $V^* \times V^*$ is the unique maximum clique of $G$ 
and will correspond to the unique solution of \eref{clique relaxation}
with probability tending exponentially  to $1$ as $N \ra \infty$.
\end{theorem}

\begin{proof}
Consider the matrix $W$ constructed as in ($\omega_1$)--($\omega_6$) with $\gamma = -p/(1 - p).$
By Theorem~\ref{thm:optimalityclique}, $X^*$ is the unique optimum if
$$
   \|W\| < 1 \;\;\; \mbox{and} \;\;\;
        p_j < n \mbox{ for all } j \in V-V^*
$$
We first show that $\|W\| < 1$ with probability tending exponentially to 1  as $N \ra \infty$ in the case that $n = \Omega(\sqrt{N})$.
We write $W = W_1 + W_2 + W_3 + W_4 + W_5$, where each of the five terms is defined as follows.

We first define $W_1$.
For cases ($\omega_2$) and ($\omega_4$), choose $W_1(i,j) = W(i,j)$.
For cases ($\omega_5$) and ($\omega_6$), take $W_1(i,j) = -p/((1-p)n)$.
For case ($\omega_1$), choose $W_1(i,j)$ randomly such that $W_1(i,j)$ is equal to $1/n$ with probability $p$ and equal to $-p/((1-p)n)$ otherwise.
Similarly, in case ($\omega_3$), take $W_1(i,i)$ to be equal to $1/n$ with probability $p$ and equal to $-p/((1-p)n)$ otherwise.
By construction, each entry of $W_1$ is an independent random 
variable with the distribution
\[
    W_1(i,j) = \left\{  \begin{array}{ll}
                        1/n         & \mbox{with probability } p, \\
                        -p/((1-p)n) & \mbox{with probability } 1-p.
                    \end{array}
            \right.
\]
Therefore, applying Lemma~\ref{Furedi-Komlos} shows that there exists constant $c_1 > 0$ such that
\begin{equation}
    \label{eqn: W1 bound}
    \|W_1\| \le 3 \left(\frac{p}{1-p} \right)^{1/2} \frac{\sqrt{N}}{n}
\end{equation}
with probability at least to $1 - \exp(c_1 N^{1/6})$ for some constant $c_1 > 0$.

Next, $W_2$ is the correction matrix to $W_1$ in case ($\omega_1$).
That is, $W_2(i,j)$ is chosen such that
\[
    W_2(i,j) + W_1(i,j) = W(i,j) = 0
\]
for all $(i,j) \in V^* \times V^*$ and is zero everywhere else.
As before, applying Lemma~\ref{Furedi-Komlos} shows that 
\begin{equation}
    \label{eqn: W2 bound}
    \|W_2\| \le 3 \left(\frac{p}{1-p} \right)^{1/2} \frac{1}{\sqrt{n}}
\end{equation}
with probability at least $1 - \exp(c_1 n^{1/6})$.
Similarly, $W_3$ is the correction to $W_3$ in case ($\omega_3$), that is
\[
    W_3(i,i) = W(i,i) - W_1(i,i)
\]
for all $i \in V - V^*$ and all other entries are equal 0.
Therefore, $W_3$ is a diagonal matrix with diagonal entries bounded by $2/n$. It follows that
\begin{equation}
    \label{eqn: W3 bound}
    \|W_3\| \le \frac{2}{n}.
\end{equation}
Finally, $W_4$ and $W_5$ are the corrections for cases ($\omega_5$) and ($\omega_6$) respectively.  These are exactly of the form $(A-\tilde A)/n$ as in 
Theorem~\ref{thm:AtildeA}, in which $N$ in the theorem stands for
$N-n$ in the present context.
Examining each term of \eref{eq:bigprob} shows that in the case
$n=\Omega(N^{1/2})$, the probability on the right-hand side is
the form $1-c\exp(-kN^{c_2})$. 
It follows that there exists constant $\alpha_4 > 0$ such that
\[	
	\|W_4\|^2 \le \|W_4\|^2_F < \alpha_4^2 N n^{-2}
\]
with probability tending exponentially to 1 as $N \ra \infty$.
Moreover, since Condition F is satisfied in this case, $p_j < n$ for all $j \in V - V^*$.
Notice that, by symmetry, $W_4 = W_5^T$.
Thus, since each of $W_1, W_2,\dots, W_5$ is bounded by an arbitrarily small 
constant if $n = \Omega(\sqrt{N})$, there exists constant $\alpha > 0$ such that $\|W\| < 1$ with probability
tending exponentially to 1 as $N \ra \infty$ as required.
\end{proof}

\section{Maximum Edge Biclique}
\label{sec:maxbiclique}

Consider a bipartite graph $G = ((U,V), E)$ where $|U| = M$, $|V| = N$.
The adjacency matrix of a biclique $H$ of $G$ is 
rank-one matrix $X\in \R^{M\times N}$.  This matrix $X$ has
the property that $X_{ij} = 0$ for all $i \in U, j\in V$
such that $(i,j) \notin E$. It follows that a biclique of $G$ of size $mn$ can be found (if one exists) by solving the
rank minimization problem
\begin{align}
        \min \; & \rank(X)      \notag        \\
        \st  \; & \sum_{i \in U} \sum_{j\in V} X_{ij} \ge mn,    \label{constraint i} \\
                & X_{ij} = 0 \;\; \forall \; (i,j) \in (U\times V)-E,      \label{constraint ii} \\
                & X \in [0,1]^{U \times V}. \label{constraint iii}
\end{align}
A rank-one solution $X^*$ to this problem corresponds to the adjacency matrix of
a biclique of $G$ containing at least $mn$ edges.
As with the maximum clique problem, this rank minimization problem is still NP-hard. As before, we underestimate $\rank(X)$ with $\|X\|_*$.
We obtain the following convex optimization problem:
\begin{align}	\label{biclique relaxation}
	\begin{array}{rl}
   		\min 	& \|X\|_* \\
		\st  	& \sum_{i \in V} \sum_{j \in V} X_{ij} \ge mn, \\
            		& X_{ij} = 0 \;\; \mbox{if } (i,j) \notin E. 
	\end{array}
\end{align}
Using the Karush-Kuhn-Tucker conditions, we derive conditions for which the adjacency matrix of a graph comprising a biclique of $G$ is optimal for this relaxation.
Indeed, the following is an immediate consequence (essentially a
restatement) of Theorem~\ref{thm:optimality1}.

\begin{theorem}
\label{thm: biclique optimality and uniqueness}
Let $U^*\times V^*$ be the vertex set of a biclique in $G$ in which 
$|U^*|=m$ and $|V^*|=n$.
Let $\bar\u\in \R^M$ be the characteristic vector of $U^*$, and let
$\bar\v\in\R^N$ be the characteristic vector of $V^*$.
Let $X^*=\bar\u\bar\v^T$.  
(Clearly $X^*$ is feasible for
\eref{biclique relaxation}).
Let $E=E(G)$ and let $\tilde E$ be its
complement.
  Suppose also that there
exist $W\in\R^{M\times N}$, $\lambda \in\R^{M\times N}$ and
$\mu\in\R_+$ such that
$W \bar{\v} = \bz $, $\bar{\u}^T W =\bz$, $\|W \| \le 1$ and
\begin{align}
    \label{biclique KKT}
    \frac{\bar{\u}\bar{\v}^T}{\sqrt{mn}} + W &= \mu \e \e^T
                                        + \sum_{\substack{(i,j) \in \tilde E }} \lambda_{ij} e_i e_j^T.
\end{align}
Then $X^*$ is an optimal solution to \eref{biclique relaxation}.  Moreover,
$G$ does not contain any biclique with more than $mn$ edges.
Furthermore, if $\Vert W\Vert<1$
and $\mu > 0$,
then $X^*$ is the unique optimizer of \eref{biclique relaxation}
and $U^*\times V^*$ is the unique optimal biclique.
\end{theorem}

In the next two subsections, we consider two scenarios for how to construct
a bipartite graph $G$ and biclique that satisfy the conditions of the
theorem.

In both scenarios,
we will take $\mu = 1/\sqrt{mn}$ and consider  $W$ and $\lambda$ defined according to the following cases.
\begin{itemize}
    \item[($\psi_1$)]
        For $(i,j) \in U^*\times V^*$, taking $W_{ij} = 0$ and $\lambda_{ij} =0$ ensures the $ij$-entries of both sides of (\ref{biclique KKT})
        are equal to $1/\sqrt{mn}$.
    \item[($\psi_2$)]
        For $(i,j) \in E - (U^* \times V^*)$, we take $W_{ij} = 1/\sqrt{mn}$ and $\lambda_{ij} = 0$. Again, the $ij$-entries of both sides of
        (\ref{biclique KKT}) are equal to $1/\sqrt{mn}$.
    \item[($\psi_3$)]
        For $(i,j) \notin E$ such that $i \notin U^*$ and $j \notin V^*$, we select $ W_{ij} = - \gamma/\sqrt{mn}$ and $\lambda_{ij} = -(1 + \gamma)/\sqrt{mn}$
        where $\gamma$ will be defined below.
        In this case, the $ij$-entries of each side of (\ref{biclique KKT}) are 0.
    \item[($\psi_4$)]
        For $(i,j) \notin E$ such that $i\notin U^*$ and $j\in V^*$, we choose
        \[
            W_{ij} = - \frac{p_i}{(n-p_i)\sqrt{mn}}
            \mbox{ and } \lambda_{ij} = \frac{1}{\sqrt{mn}}\left( \frac{-p_i}{n-p_i} -1 \right)
        \]
        where $p_i$ is equal to the number of edges with left endpoint equal to $i$ and right endpoint in $V^*$.
        Note that if $n= p_i$ then $i$ is connected to every vertex of $V^*$ and thus the KKT condition cannot possibly be
        satisfied.
        If $p_i < n$, both sides of (\ref{biclique KKT}) are equal to $-p_i / ((n-p_i)\sqrt{mn})$.
    \item[($\psi_5$)]
        For $(i,j) \notin E$ such that $i\in U^*$ and $j \notin V^*$, we choose
        \[
            W_{ij} = - \frac{q_j}{(m-q_j)\sqrt{mn}}
            \mbox{ and } \lambda_{ij} = \frac{1}{\sqrt{mn}}\left( \frac{-q_j}{m-q_j} -1 \right)
        \]
        where $q_j$ is equal to the number of edges with right endpoint equal to $j$ and left endpoint in $U^*$.
        As before, this is appropriate only if $q_j < m$.
\end{itemize}

We next check that $W$ satisfies the requirements for $\phi$ to be a subgradient of $\bar{\u}\bar{\v}^T$:
$W\bar{\v} = \bz$, $W^T \bar{\u} = \bz$, and $\|W\| \le 1$.
To show that $W \bar{\v} = \bz$, choose row $i$ of $W$ and consider $W(i,:)\bar{\v} = \sum_{j \in V^*} W_{ij}$.
If $i \in U^*$ then $W_{ij} = 0$ for all $j\in V^*$, so $W(i,:)\bar{\v} = 0$.
In the case $i \notin U^*$, consider each $j\in V^*$. If $(i,j) \in E$ then, by Case 2, $W_{ij} = 1/\sqrt{mn}$.
There are $p_i$ such entries, with sum $p_i/\sqrt{mn}$.
If $(i,j) \notin E$, then $W_{ij} = -p_i/((n-p_i)\sqrt{mn})$.
There are $n-p_i$ such entries, with sum $-p_i/\sqrt{mn}$.
It follows that $W(i,:)\bar{\v} = 0$ as required.

The proof that $W^T\bar{\u} = \bz$ follows is symmetric. It remains to determine which bipartite graphs $G$ yield $W$ as defined
above such that $\|W\| < 1.$
As in the maximum clique case, we present two different analyses.

%\subsection{Optimality Conditions}
%\label{sec:maxbicliqueopt}

%\input{biclique-optimality-conds.tex}

\subsection{The Adversarial Case}
\label{sec:maxbicliqueadver}

Suppose that the edge set of the bipartite graph $G = ((U,V), E)$ is generated as follows. We first add a biclique $U^* \times V^*$ with $|U^*| = m$, $|V^*| = n$.
Then, as in the adversarial case for the maximum clique problem, an adversary is allowed to add a number of the remaining $|U||V| - mn$ potential edges to the graph.
We will show that, under certain conditions, our adversary can add up to $O(mn)$ edges to the graph and $U^* \times V^*$ will still be a maximum edge
biclique of $G$.

We make the following assumptions on the structure of $G$:
\begin{enumerate}
	\item
       		$G$ contains at most $r$ edges aside from those of the optimal biclique.
        	\item
		Each vertex of $V-V^*$ is adjacent to at most $\alpha m$ vertices of $U^*$ for some $\alpha \in (0,1)$.	
        	\item
        		Each vertex of $U - U^*$ is adjacent to at most $\beta n$ vertices of $V^*$ for some $\beta \in (0,1)$.
\end{enumerate}
Consider $W$ as defined by $(\psi_1)$-$(\psi_5)$ with $\gamma = 0$.
As before, we use the bound $\|W\| \le \|W\|_F$.
Notice that at most $r$ entries of $W(U-U^*, V-V^*)$ are equal to $1/\sqrt{mn}$ and the remainder are equal to 0. Therefore,
$$
	\|W(U- U^*, V-V^*)\|_F^2 \le \frac{r}{mn}.
$$
Moreover, for each $j \in V-V^*$, $q_j \le \alpha m$.
It follows that
\begin{align*}
	\|W(U^*, V-V^*)\|_F^2 &= \sum_{v \in V-V*} \left(\frac{q_v}{mn} + (m -q_v) \frac{q_v^2}{mn(m- q_v)^2} \right)\\
		&= \sum_{v\in V^*} \frac{q_v}{mn} \left( 1 + \frac{q_v}{m-q_v} \right) \\
		&\le \sum_{v \in V^*} \frac{q_v}{mn} \left(1 + \frac{\alpha}{1-\alpha} \right) \\
		&= \sum_{v \in V^*} \frac{q_v}{mn(1-\alpha)} \le \frac{r}{mn(1-\alpha)}.
\end{align*}
Similarly,
$$
	\| W(U-U^*, V^*) \|_F^2 \le \frac{r}{(1-\beta) mn}.
$$
Therefore, $\|W\| < 1$ if
$$
	r \left( 1 + \frac{1}{1-\alpha} + \frac{1}{1- \beta} \right) < mn.
$$
Thus, the graph can contain up to $O(mn)$ diversionary edges, yet the optimality and uniqueness conditions given by Theorem~\ref{thm: biclique optimality and uniqueness}
are still satisfied.
This result is the best possible up to constants for the same reasons explained at the end of Section~\ref{sec:maxcliqueadver}.

\subsection{The Random Case}
\label{sec:maxbicliquerand}

Let $y,z$ be fixed positive scalars.
Let $U, V$ be two disjoint vertex sets with $|V| = N$ and
$|U| = \lceil yN\rceil $. Consider 
$U^* \subseteq U$ and $V^* \subseteq V$ such that 
$|V^*| = n$ and
$|U^*| =  m=\lceil zn \rceil$.
Suppose the edges of the bipartite graph $G = ((U,V), E)$ are determined as follows:
\begin{itemize}
    \item[($\beta_1$)]
        For all $(i,j) \in U^*\times V^*$, $(i,j) \in E$.
    \item[($\beta_2$)]
        For each of the remaining potential edges $(i,j) \in U \times V$, we add edge $(i,j)$ to $E$ with probability $p$ (independently).
\end{itemize}
Notice $G$ contains the biclique $(U^*,V^*)$.
As in the maximum clique problem, if $n = \Omega(\sqrt{N})$ and $G$ is constructed as in $(\beta_1)$, $(\beta_2)$
then $U^* \times V^*$ is optimal for the convex problem (\ref{biclique relaxation}).
We have the following theorem.

\begin{theorem}
    \label{main thm: biclique}
There exists $\alpha > 0$ depending on $p$, $y$, $z$ such that for each bipartite graph $G$ constructed via $(\beta_1)$, $(\beta_2)$ with $n \ge \alpha \sqrt{N}$
    the biclique defined by $U^* \times V^*$ is maximum edge biclique of $G$ with probability tending exponentially to $1$ as $N \ra \infty$ and is found as the unique solution to the convex relaxation \eref{biclique relaxation}.
\end{theorem}

Let $W$ be constructed as in $(\psi_1)$--$(\psi_5)$ with $\gamma = - p/(1-p)$.
Then $X^* = \bar{\u}\bar{\v}^T$ is the unique optimal solution of $(\ref{biclique relaxation})$
if
$$
    \|W\| < 1, \;\;\; q_j < \lceil z n\rceil  \; \forall\, j\in V-V^*, \;\; \mbox{and } p_j < n \; \forall \; j \in U-U^*.
$$
To prove that $\|W\| < 1$ with high probability as $N \ra \infty$ in the case that $n = \Omega(\sqrt{N})$, we write
$$
    W = W_1 + W_2 + W_3 + W_4
$$
where each of the summands is defined as follows.
We first define $W_1$.  
If $(i,j) \in U^* \times V^*$, then we set $W_1(i,j) = 1/\sqrt{mn}$
with probability $p$ and equal to $\gamma/\sqrt{mn}$ with probability $(1-p)$.
For $(i,j) \in (U\times V) - (U^* \times V^*)$, we set $W_1(i,j) = 1/\sqrt{mn}$
if $(i,j) \in E$ and set $W_1(i,j) = \gamma/\sqrt{mn}$ otherwise.
In order to bound $\|W_1\|$, we will use the following
Theorem~\ref{geman tail bound} to conclude that
$\Vert W_1\Vert \le \alpha \sqrt{N}/\sqrt{mn}$.  Since $\sqrt{mn}$ 
equals $\sqrt{\lceil zn\rceil n}$ and hence
is proportional
to $n$, we see that $\Vert W_1\Vert \le {\rm const}$ with probability
exponentially close to 1 provided $n=\Omega(\sqrt{N})$.

Next, set $W_2$ to be the correction matrix for $W_1$ for $U^* \times V^*$, that is,
$$
    W_2(i,j) = \left\{  \begin{array}{ll}
                    -W_1(i,j) & \mbox{if } (i,j) \in U^* \times V^* \\
                    0       &\mbox{otherwise,}
                \end{array}
            \right.
$$
Again, by Theorem~\ref{geman tail bound} we conclude that
$$
    \|W_2\| \le \alpha \frac{1}{\sqrt{n}}
$$
with probability at least $1- c_1' \exp(-c_2'n^{c_3'})$ for some $c_1', c_2', c_3' > 0$.

It remains to derive bounds for $\|W_3\|$ and $\|W_4\|$. Notice that the construction of
$W(U^*, V-V^*)$ and $W(U- U^*, V^*)$ is identical to that in Case~$(\omega_5)$ for the maximum clique problem.
Thus, we can again apply Theorem~\ref{thm:AtildeA}, first to $W_3$ (in which case
$(n,N)$ in the theorem stand for $(\lceil zn\rceil ,N-n)$) and second to
$W_4^T$ (in which case $(n,N)$ in the theorem stand for $(n, \lceil yN\rceil -
\lceil zn\rceil)$ to conclude that $\Vert W_3\Vert$ and $\Vert W_4\Vert$ are
both strictly bounded above by constants provided $n=\Omega(\sqrt{N})$ with probability tending
to $1$ exponentially fast. Moreover, as before, Condition F is satisfied in this case and thus $q_j < \lceil zn \rceil$ for all $j \in V-V^*$
and $p_j < n$ for all $j \in U - U^*$ as required.
\qed

\section{Conclusions}

We have shown that the maximum clique and maximum biclique problems
can be solved in polynomial time using nuclear norm minimization, a
technique recently proposed in the compressive sensing literature,
provided that the input graph consists of a single clique or biclique plus
diversionary edges.  The spectral technique used
by Alon et al.\ \cite{Alon} for the planted clique problem
has been extended
to other problems; see, e.g., 
McSherry \cite{959929}.  It would be interesting to extend the
nuclear norm approach to other NP-hard problems as well.

\section{Acknowledgements}

We received helpful comments from our colleagues C.~Swamy and
D.~Chakrabarty.

%******************************************************************************
% There are a number of bib. styles around, ask your advisor
% as to which one they prefer, or use this one.

\bibliographystyle{plain}
\bibliography{rankmin}

% This command, and the includes are optional - but only
% if you have no appendices
%******************************************************************************

\end{document}